\documentclass{llncs}

\usepackage{amsmath,amssymb}
\usepackage{enumerate,enumitem}
\usepackage{graphicx}
\usepackage{subfigure}
\usepackage{url}
\usepackage{times,citesort}

\let\doendproof\endproof
\renewcommand\endproof{~\hfill\qed\doendproof}

\title{Equilateral L-Contact Graphs}
\author{Steven Chaplick\inst{1}\thanks{Supported by NSERC, and ESF GraDR EUROGIGA grant as project GACR GIG/11/E023.}
\and 
Stephen G.~Kobourov\inst{2}\thanks{Research funded in part by NSF grants CCF-1115971 and DEB-1053573.}
\and 
Torsten Ueckerdt\inst{3}
}
\institute{Dept.~of Applied Mathematics, Charles University, Prague, Czech Republic
\and Dept.~of Computer Science, University of Arizona, Tucson AZ, USA \and Dept.~of Mathematics, Karlsruhe Istitute of Technology, Karlsruhe, Germany}

\begin{document}

\date{}

 \maketitle

\begin{abstract}
 We consider {\em L-graphs}, that is contact graphs of axis-aligned L-shapes in the plane, all with the same rotation. We provide several characterizations of L-graphs, drawing connections to Schnyder realizers and canonical orders of maximally planar graphs. We show that every contact system of L's can always be converted to an equivalent one with equilateral L's. This can be used to show a stronger version of a result of Thomassen, namely, that every planar graph can be represented as a contact system of square-based cuboids.
 
We also study a slightly more restricted version of equilateral L-contact systems and show that these are equivalent to homothetic triangle contact representations of maximally planar graphs. We believe that this new interpretation of the problem might allow for efficient algorithms to find homothetic triangle contact representations, that do not use Schramm's monster packing theorem.
\end{abstract}

\section{Introduction}\label{sec:introduction}
A \emph{contact graph} is a graph whose vertices are represented by geometric objects (such as curves, line segments, or polygons), and edges correspond to two objects touching in some specified fashion. There is a large body of work about representing planar graphs as contact graphs. An early result is Koebe's 1936 theorem~\cite{Koebe36} that all planar graphs can be represented by touching disks.

In 1990 Schnyder showed that maximally planar graphs contain rich combinatorial structure~\cite{s-epgg-90}. With an angle labeling and a corresponding edge labeling, Schnyder shows that maximally planar graphs can be decomposed into three edge disjoint spanning trees. This combinatorial structure, called Schnyder realizer, can be transformed into a geometric structure to produce a straight-line crossing-free planar drawing of the graph with vertex coordinates on the integer grid. While Schnyder realizers were defined for maximally planar graphs~\cite{s-epgg-90}, the notion generalizes to $3$-connected planar graphs~\cite{f-lspg-04}. Fusy's transversal structures~\cite{Fusy09} for irreducible triangulations of the 4-gon also provide combinatorial structure that can be used to obtain geometric results. Later, de~Fraysseix {\em et al.}~\cite{FraysseixTContact} show how to use Schnyder realizer to produce a representation of planar graphs as T-contact graphs (vertices are axis-aligned T's and edges correspond to point contact between T's) and triangle contact graphs.

Recently, a similar combinatorial structure, called \emph{edge labeling}, was identified for the class of planar Laman graphs, and  this was used to produce a representation of such graphs as L-contact graphs, with L-shapes in all four rotations~\cite{full}.
 Planar Laman graphs contain several large classes of planar graphs (e.g., series-parallel graphs, outer-planar graphs, planar 2-trees) and are also of interest in structural mechanics, chemistry and physics, due to their connection to rigidity theorys~\cite{hors+-pmrgpt-05}. This dates back to the 1970's~\cite{Laman}. A system of fixed-length bars and flexible joints connecting them is minimally rigid if it becomes flexible once any bar is removed; planar Laman graphs correspond to rigid planar bar-and-joint systems~\cite{hors+-pmrgpt-05}.

Planar bipartite graphs can be represented by axis-aligned segment
contacts~\cite{CzyzowiczKU98,fop-rpgs-91,rt-rplbopg-86}.
Triangle-free planar graphs can be represented via
contacts of segments with only three slopes~\cite{CastroCDMN02}. They can also be represented by contact axis-aligned line segments, L-shapes, and $\Gamma$-shapes~\cite{ChaplickU12}.
Furthermore, every $4$-connected $3$-colorable planar graph and every $4$-colored planar graph without an induced $C_4$ using four colors can be represented as the contact graph of segments~\cite{fo-rcis-07}. More generally, planar Laman graphs can be represented with contacts of segments with arbitrary number of slopes and every contact graph of segments is a subgraph of a planar Laman graph~\cite{A+11}.

Planar graphs have also been considered as intersection graphs of geometric objects. One major result is the proof of Scheinerman's conjecture that all planar graphs are intersection graphs of line segments in the plane~\cite{Chalopin:2009}. Recently the \emph{$k$-bend Vertex intersection graphs of Paths in Grids ($B_k$-VPG)}were introduced and it was shown that planar graphs are $B_3$-VPG~\cite{Asinowski2012}. It was recently shown that planar graphs are $B_2$-VPG~\cite{ChaplickU12}, where the authors also conjectured that all planar graphs are a intersection graphs of one fixed rotation of axis-aligned L-shapes (a special case of $B_1$-VPG). 

In the 3D case Thomassen~\cite{Thomassen86} shows that any planar graph has a proper contact representation by touching cuboids (axis-alligned boxes). Felsner and Francis~\cite{Felsner11} show that any planar graph has a (not necessarily proper) representation by touching cubes. In a \emph{proper contact representation of cuboids} contacts must always have non-zero area and \emph{cubes} are special cuboids where all sides have the same length.
Recently Bremner {\em et al.}~\cite{cubes12} showed that deciding whether a graph admits a proper contact representation by unit cubes is NP-Complete. They also show that with cubes of varying sizes one can find proper contact representation for some planar graph classes such as partial planar $3$-trees. Finally, they describe two new proofs of Thomassen's result: one based on canonical orders~\cite{fpp-hdpgg-90} of de Fraysseix, Pach and Pollack~\cite{fpp-hdpgg-90} and the other based on Schnyder's realizers~\cite{s-epgg-90}. 

\medskip
\noindent{\bf Our Contributions:}
In this paper we consider contact graphs of L-shapes in only one fixed rotation, so-called L-graphs.
In Section~\ref{sec:preliminaries} we briefly review Schnyder realizers, T-contact representations, triangle contact representations, and canonical orders.
In Section~\ref{sec:characterization} we characterize L-graphs in terms of canonical orders, Schnyder realizers, and edge labelings.
We also show how to recognize L-graphs in quadratic time.
In Section~\ref{sec:equilateral} we show that every L-representation has an equivalent one with only equilateral L-shapes. Using this we strengthen the result of Thomassen~\cite{Thomassen86} and Bremner {\em et al.}~\cite{cubes12}, by showing that every planar graph admits a proper contact representation with square-based cuboids.
Finally, we consider a special class of equilateral L-representations, drawing connections to homothetic triangle contact representations of maximally planar graphs and contact representations with cubes.
The question whether every planar graph has a proper contact representation by cubes remains tantalizingly open.

\section{Preliminaries}\label{sec:preliminaries}

Schnyder realizers for maximally planar graph were originally described in 1990~\cite{s-epgg-90} and have played a central role in numerous problems for planar graphs.

\begin{definition}[\cite{s-epgg-90}]
 Let $G =(V,E)$ be a maximally planar graph with a fixed plane embedding. Let $v_1,v_2,v_n$ be the outer vertices in clockwise order. A \emph{Schnyder realizer of $G$} is an orientation and coloring of the inner edges of $G$ with colors $1$ (red), $2$ (blue) and $n$ (green), such that:
 \begin{enumerate}[label =(\roman*)]
  \item Around every inner vertex $v$ in clockwise order there is one outgoing red edge, a possibly empty set of incoming green edges, one outgoing blue edge, a possibly empty set of incoming red edges, one outgoing green edge, a possibly empty set of incoming blue edges.
  \item All inner edges at outer vertices are incoming and edges at $v_1$ are colored red, edges at $v_2$ are colored blue, edges at $v_n$ are colored green.
 \end{enumerate}
\end{definition}

Schnyder realizers have several useful properties; see Fig.~\ref{fig:Schnyder}. For example, if $S_1$, $S_2$ and $S_n$ are the sets of red, blue and green edges, then for $i=1,2,n$ we have that $S_i$ is a directed tree spanning all inner vertices plus $v_i$, where each edge is oriented towards $v_i$. This way the orientation of edges can be recovered from their coloring and hence we denote a Schnyder realizer simply by the triple $(S_1,S_2,S_n)$.
For $i=1,2,n$ let $S_i^{-1}$ be the set $S_i$ with the orientation of every edge reversed. It is well-known that for every Schnyder realizer $S_1 \cup S_2 \cup S_n^{-1}$ is an acyclic set of directed edges.

\begin{figure}[t!]
 \centering
 \subfigure[]{
  \includegraphics{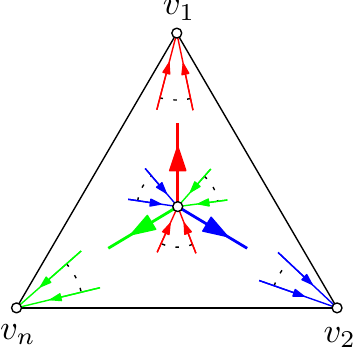}
  \label{fig:Schnyder-rules}
 }
 \subfigure[]{
  \includegraphics{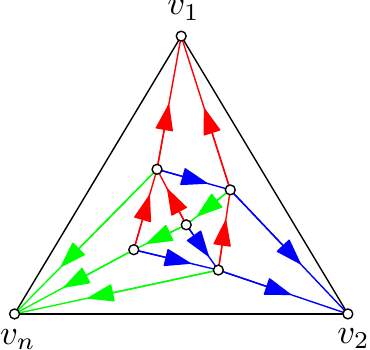}
  \label{fig:Schnyder-example}
 }
 \subfigure[]{
  \includegraphics{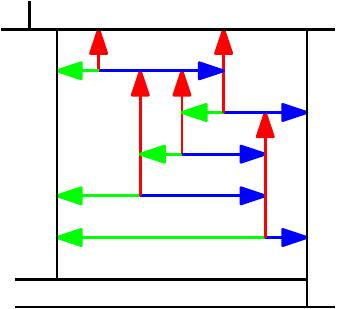}
  \label{fig:Schnyder-Ts}
 }
 \subfigure[]{
  \includegraphics{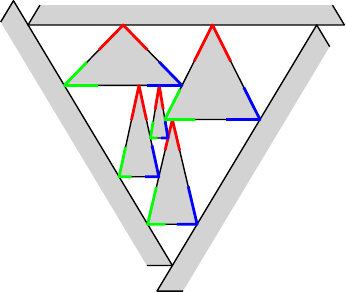}
  \label{fig:Schnyder-triangles}
 }
 \subfigure[]{
  \includegraphics{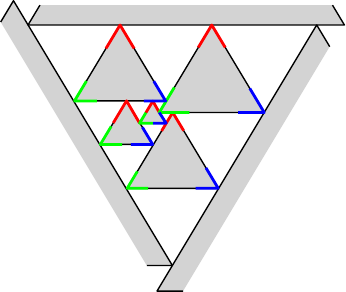}
  \label{fig:Schnyder-homothetics}
 }
 \subfigure[]{
  \includegraphics{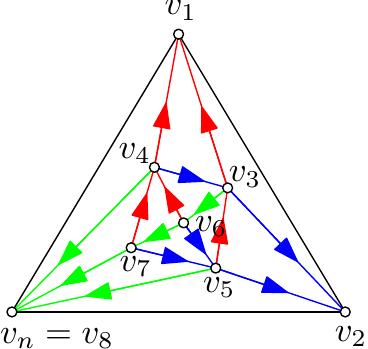}
  \label{fig:canonical-example}
 }
 \caption{\subref{fig:Schnyder-rules} The Schnyder rules for inner and outer vertices. \subref{fig:Schnyder-example} A maximally planar graph $G$ with a Schnyder realizer $(S_1,S_2,S_n)$. \subref{fig:Schnyder-Ts} A T-contact representation of $G$ w.r.t. $(S_1,S_2,S_n)$. \subref{fig:Schnyder-triangles} A triangle contact representation of $G$ w.r.t. $(S_1,S_2,S_n)$. \subref{fig:Schnyder-homothetics} A homothetic triangle representation of $G$ w.r.t. $(S_1,S_2,S_n)$. \subref{fig:canonical-example} A canonical order of $G$ w.r.t. $S_1,S_2$.}
 \label{fig:Schnyder}
\end{figure}

Schnyder realizers are often used show that planar graphs admit certain contact representations. In a \emph{T-contact representation} of a maximally planar graph $G = (V,E)$ the vertices are assigned to interior disjoint axis-aligned upside down T-shapes, so that two T-shapes touch in a point if and only if the corresponding vertices are joined by an edge in $G$. For a vertex $v \in V$ let $T_v$ be the corresponding T-shape. From every T-contact representation we get a Schnyder realizer by coloring an edge $uv$ red (respectively blue and green) if the top (respectively left and right) endpoint of $T_u$ is contained in $T_v$; see Fig.~\ref{fig:Schnyder-Ts}.

Similarly to T-contact representations, de Fraysseix \textit{et al.}~\cite{FraysseixTContact} consider triangle contact representations. In a \emph{triangle contact representation} of a maximally planar graph $G = (V,E)$ the vertices are assigned to interior disjoint triangles, so that two triangles touch in a point if and only if the corresponding vertices are joined by an edge in $G$. We can indeed assume w.l.o.g. all triangles are isosceles with horizontal bases and the tip above. For a vertex $v \in V$ let $\Delta_v$ be the corresponding triangle. We again get a Schnyder realizer by coloring an edge $uv$ red (respectively blue and green) if the top (respectively left and right) corner of $\Delta_u$ is contained in $\Delta_v$; see Fig.~\ref{fig:Schnyder-triangles}.

\begin{theorem}[\cite{FraysseixTContact}]
 Let $G$ be a maximally planar graph with a fixed embedding. Then:
 \begin{itemize}
  \item Every T-contact representation defines a Schnyder realizer and vice versa.
  \item Every triangle contact representation defines a Schnyder realizer and vice versa.
 \end{itemize} 
\end{theorem}

A \emph{homothetic triangle representation} is a triangle contact representation in which all triangles are homothetics. It has been noticed by Gon{\c{c}}alves, L{\'e}v{\^e}que and Pinlou~\cite{gonccalves2012triangle}, that a result of Schramm~\cite{schramm2007combinatorically} implies the following.

\begin{theorem}[\cite{gonccalves2012triangle}]\label{thm:homothetic-triangles}
 Every $4$-connected maximally planar graph admits a homothetic triangle representation.
\end{theorem}

Canonical orders were first introduced by De Fraysseix, Pach and Pollack in 1990~\cite{fpp-hdpgg-90}. For maximally planar graphs Schnyder realizers and canonical orders are very closely related, as shown in Lemma~\ref{lem:Schnyder-canonical} below.

\begin{definition}[\cite{fpp-hdpgg-90}]\label{def:canonical-order}
 Let $G = (V,E)$ be a biconnected planar graph with a fixed embedding and some distinguished outer edge $v_1v_2$. A \emph{canonical order of $G$} is a permutation $(v_1,v_2,v_3,\ldots,v_n)$ of the vertices of $G$, such that:
 \begin{enumerate}[label =(\roman*)]
  \item For each $i \geq 3$ the induced subgraph $G_i$ of $G$ on $\{v_1,\ldots,v_i\}$ is biconnected, and the boundary of its outer face is a cycle $C_i$ containing the edge $v_1v_2$.
  \item For each $i \geq 4$ the vertex $v_i$ lies in the outer face of $G_{i-1}$, and its neighbors in $G_{i-1}$ form a subpath of $C_i \setminus v_1v_2$.
 \end{enumerate}
 The outer edge $v_1v_2$ of $G$ is then called the \emph{base edge of the canonical order}.
\end{definition}

\begin{lemma}\label{lem:Schnyder-canonical}
 If $G$ is a maximally planar graph with Schnyder realizer $(S_1,S_2,S_n)$, then every topological ordering of $S_1 \cup S_2 \cup S_n^{-1}$ defines a canonical order of $G$. Moreover, every canonical order of $G$ is a topological order of $S_1 \cup S_2 \cup S_n^{-1}$ for some Schnyder realizer $(S_1,S_2,S_n)$.
\end{lemma}

We call a canonical order that is a topological order of $S_1\cup S_2 \cup S_n^{-1}$ a \emph{canonical order w.r.t. $S_1,S_2$}. See Fig.~\ref{fig:canonical-example} for an example. Note that the same Schnyder realizer may give rise to several canonical orders as for example swapping the order of $v_4$ and $v_5$ in Fig.~\ref{fig:canonical-example} results in a different canonical order w.r.t. $S_1,S_2$.

\medskip

Another vertex order that can be defined for any graph is the so-called $k$-degenerate order. For an $n$-vertex graph $G$ and a number $k \in \mathbb{N}$ $(v_1,\ldots,v_n)$ is a \emph{$k$-degenerate order of $G$} if for each $i = 1,\ldots,n$ the vertex $v_i$ has no more than $k$ neighbors in the induced subgraph $G_{i-1}$ of $G$ on $\{v_1,\ldots,v_{i-1}\}$. A graph is \emph{$k$-degenerate} if it admits some $k$-degenerate order, and \emph{maximally $k$-degenerate} if for each $i \in \{1,\ldots,n\}$ vertex $v_i$ has exactly ${\rm min}\{i-1,k\}$ neighbors in $G_{i-1}$. A very important subclass of maximally $k$-degenerate graphs are $k$-trees. A maximally $k$-degenerate graph $G$ is a \emph{$k$-tree} if in some $k$-degenerate order of $G$ the neighborhood of $v_i$ is a clique in $G_{i-1}$, $i = 1,\ldots,n$. Equivalently, $k$-trees are exactly the inclusion-maximal graphs of tree-width $k$.

\section{Contact L-graphs: Characterization and Recognition}\label{sec:characterization}

An \emph{L-contact representation}, or \emph{L-representation} for short, of a graph $G = (V,E)$ is a set of interior disjoint axis-aligned L-shapes, one for each vertex, such that two L-shapes touch in a point if and only if the corresponding vertices in $G$ are adjacent. Unless stated otherwise we allow only one of the four possible rotations of L-shapes here. An L-representation is \emph{degenerate} if two endpoints of L-shapes or an endpoint and a bend coincide, and \emph{non-degenerate} otherwise. A graph is an \emph{L-contact graph} or simply \emph{L-graph} if it admits an L-representation. Since one can remove any contact in an L-representation by shortening one L, L-graphs are closed under taking subgraphs. Throughout this section we consider \emph{maximal L-graphs} only, that is, L-graphs (with at least two vertices), that are not proper subgraphs of another L-graph.

For a fixed L-representation we denote the L-shape corresponding to a vertex $v$ by $L_v$. The vertex for the L-shape with topmost horizontal leg and the vertex for the L-shape with rightmost vertical leg is denoted by $v_1$ and $v_2$, respectively. The edge $v_1v_2$ is called the \emph{base edge of the L-representation}. Every L-representation defines a plane embedding of the underlying L-graph $G$. Each inner face of $G$ corresponds to a rectilinear polygon whose boundary lies in the union of L-shapes for the vertices of that face. 
The L-shapes whose bends lie in at most one such rectilinear polygon correspond to the outer vertices of $G$. The maximal rectilinear path $S$ containing all bends of these L-shapes is called the \emph{outer staircase of the L-representation}. The L-shapes appear along $S$ starting with $L_{v_1}$ and ending with $L_{v_2}$ in the same order as the outer vertices of $G$ along the outer face starting with $v_1$ and ending with $v_2$; see Fig.~\ref{fig:example}.

%

\begin{figure}[t!]
 \centering
 \subfigure[]{
  \includegraphics{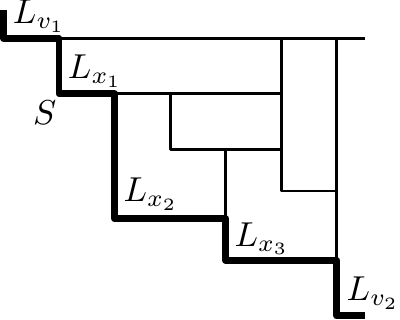}
  \label{fig:example-L}
 }
 \subfigure[]{
  \includegraphics{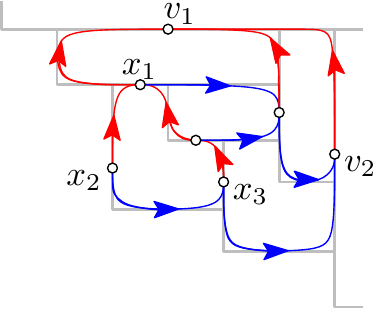}
  \label{fig:example-labeling}
 }
 \subfigure[]{
  \includegraphics{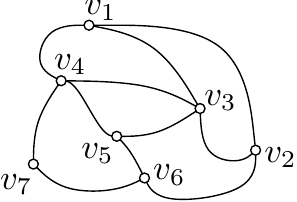}
  \label{fig:example-canonical}
 }
 \caption{\subref{fig:example-L} An L-representation with base edge $v_1v_2$ and outer staircase $S \subset L_{v_1} \cup L_{x_1} \cup L_{x_2} \cup L_{x_3} \cup L_{v_2}$ drawn thick. \subref{fig:example-labeling} The corresponding embedded L-graph with the corresponding edge labeling. \subref{fig:example-canonical} A corresponding $2$-canonical order of the graph.}
 \label{fig:example}
\end{figure}

For a maximally planar graph $G$ and a Schnyder realizer $(S_1,S_2,S_n)$ of $G$ we define $G \setminus S_n$ as the graph $(V \setminus v_n, E \setminus S_n)$.

\begin{lemma}\label{lem:L-is-2/3-Schnyder}
 For every maximal L-graph $G$ with base edge $v_1v_2$ there is a maximally planar graph $H$ with a Schnyder realizer $(S_1,S_2,S_n)$, such that $G = H \setminus S_n$.
\end{lemma}
\begin{proof}
%
 We consider any L-representation of $G$ with base edge $v_1v_2$. We introduce a T-shape $T_{v_n}$ whose vertical leg lies to the left of $L_{v_1}$ and whose horizontal leg lies below $L_{v_2}$. We obtain a T-representation by adding a left leg to every L-shape so that its endpoint touches some vertical leg but is interior disjoint from any other leg. Let $H$ be the maximally planar graph with that T-representation and $(S_1,S_2,S_n)$ be the corresponding Schnyder realizer. Then $G = H \setminus S_n$.
\end{proof}

Recall from Definition~\ref{def:canonical-order} that if $(v_1,\ldots,v_n)$ is a canonical order of some biconnected graph $G$, then for every $i \in \{3,\ldots,n\}$ the subgraph $G_i = G[v_1,\ldots,v_i]$ is also biconnected, which implies that for each $i = 3,\ldots,n$ the vertex $v_i$ has degree at least two in $G[v_1,\ldots,v_i]$. A \emph{$2$-canonical order} is a canonical order for which each $v_i$ has degree exactly two in $G_i$. In particular a $2$-canonical order is a special $2$-degenerate order of a planar graph that depends on the chosen embedding. Note that there are planar maximal $2$-degenerate graphs that admit no $2$-canonical order; see Fig.~\ref{fig:not-2-canonical-2} and~\subref{fig:not-2-canonical}.
Note also that the graph in Fig.~\ref{fig:not-2-canonical-2} admits a $2$-degenerate order in which every vertex is put into the outer face of the graph induced by vertices of smaller index.

\begin{figure}[t!]
 \centering
 \subfigure[]{
  \includegraphics{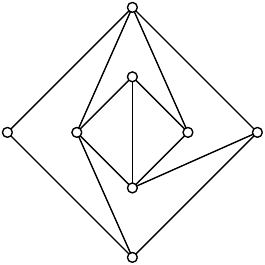}
  \label{fig:not-2-canonical-2}
 }
 \subfigure[]{
  \includegraphics{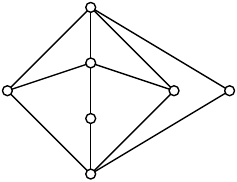}
  \label{fig:not-2-canonical}
 }
 \hspace{1em}
 \subfigure[]{
  \includegraphics{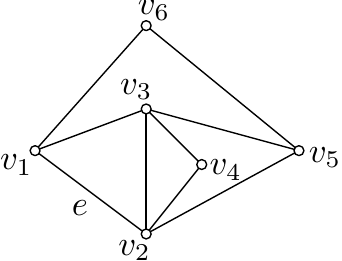}
  \label{fig:no-good-starting-edge}
 }
 \caption{\subref{fig:not-2-canonical-2},\subref{fig:not-2-canonical} Planar maximal $2$-degenerate graphs that admit no $2$-canonical order. \subref{fig:no-good-starting-edge} A graph with a $2$-canonical order with base edge $e = v_1v_2$. 
 }
 \label{fig:small-graphs}
\end{figure}

\begin{lemma}\label{lem:L-is-2-canonical}
 If a graph admits a $2$-canonical order with base edge $v_1v_2$ then it admits an L-representation with base edge $v_1v_2$. Moreover, given a $2$-canonical order an L-representation can be found in linear time.
\end{lemma}
\begin{proof}
 We use induction on the number of vertices, where the base case of just two vertices trivially holds. So let $G$ be a graph on at least three vertices.
%
%
%
 Assume that $G$ admits a $2$-canonical order and let $x$ be the last vertex in the order. Applying induction to $G \setminus x$ -- a graph with a $2$-canonical order in which both neighbors of $x$ lie on the outer face -- we obtain an L-representation of $G \setminus x$. The L-shapes for the two neighbors, $u$ and $v$, of $x$ appear on the outer staircase $S$. It is now possible to add an L-shape $L_x$, making contact with $L_u$ and $L_v$, and this way obtain an L-representation of $G$.
\end{proof}

For a graph $G$ with a fixed plane embedding and distinguished outer edge $v_1v_2$ we define an \emph{edge labeling of $G$ with base edge $v_1v_2$} to be an orientation and coloring of the edges of $G$ different from $v_1v_2$ with colors $1$ (red) and $2$ (blue), such that:
\begin{enumerate}[label = (\roman*)]
 \item Around every inner vertex $v$ in clockwise order there is one outgoing red edge, one outgoing blue edge, a possibly empty set of incoming red edges, a possibly empty set of incoming blue edges.
 \item All non-base edges at $v_1$ ($v_2$) are incoming at $v_1$ ($v_2$) and colored red (blue).
 \item Reversing all edges of color $1$ gives an acyclic graph.
\end{enumerate}

\noindent
The labeling defined above is a special case of the edge labeling in~\cite{full}, which characterizes contact L-representations with L-shapes in all four rotations.

\begin{theorem}\label{thm:characterization}
 For every graph $G$ with a plane embedding and distinguished outer edge $v_1v_2$ the following are equivalent:
 \begin{enumerate}[label = (C\arabic*)]
  \item $G$ admits an L-representation with base edge $v_1v_2$.\label{item:L-representation}
  \item $G = H \setminus S_n$ for some maximally planar graph $H$ and Schnyder realizer $(S_1,S_2,S_n)$.\label{item:2/3-Schnyder}
  \item $G$ admits an edge labeling with base edge $v_1v_2$.\label{item:edge-labeling}
  \item $G$ admits a $2$-canonical order with base edge $v_1v_2$.\label{item:2-canonical}
 \end{enumerate}
\end{theorem}
\begin{proof}
 \begin{itemize}[itemindent = 45pt]
  \item[\ref{item:L-representation} $\Longrightarrow$ \ref{item:2/3-Schnyder}:] This is Lemma~\ref{lem:L-is-2/3-Schnyder}.
  
  \item[\ref{item:2/3-Schnyder} $\Longrightarrow$ \ref{item:edge-labeling}:] Follows immediately from the definition of a Schnyder realizer.
  
  \item[\ref{item:edge-labeling} $\Longrightarrow$ \ref{item:2-canonical}:] Consider an orientation and coloring of $E(G) \setminus v_1v_2$ with the above properties. We do induction on the number of vertices of $G$. For $|V(G)| = 2$ there is nothing to show. For $|V(G)| \geq 3$ consider the path $P=x_0,x_1,\ldots,x_k,x_{k+1}$ on the outer face of $G$ not containing the edge $v_1v_2$, where $x_0 = v_1$ and $x_{k+1}= v_2$. Since the edges $x_0x_1$ and $x_kx_{k+1}$ are oriented towards $x_0$ and $x_{k+1}$, respectively, for some $i \in \{1,\ldots,k\}$ the edges $x_{i-1}x_i$ and $x_ix_{i+1}$ are outgoing at $x_i$. Since every vertex different from $v_1$ and $v_2$ has one outgoing red and one outgoing blue edge, we find a directed red path from $x_i$ to $v_1$ and a directed blue path from $x_i$ to $v_2$. No vertex $v \neq x_i$ lies on both these paths, since otherwise we would have a directed after reversing all red edges. It follows that $x_{i-1}x_i$ is colored red and $x_ix_{i+1}$ is blue. From the local 
coloring around $x_i$ we see that $x_i$ has no incoming edge. Applying induction to $G \setminus x_i$ we obtain a $2$-canonical order of $G \setminus x_i$ and putting $x_i$ at the end of this order gives a $2$-canonical order of $G$.
  
  \item[\ref{item:2-canonical} $\Longrightarrow$ \ref{item:L-representation}:] This is Lemma~\ref{lem:L-is-2-canonical}.
 \end{itemize}
\end{proof}


The remainder of this section deals with the recognition problem of maximal L-graphs. 
From Theorem \ref{thm:characterization}, every maximal L-graph is necessarily 2-degenerate and planar. Moreover, both planarity \cite{Hopcroft1974} and 2-degeneracy can be tested in linear time. 
For the maximal 2-degeneracy test, we simply iteratively remove a vertex of smallest degree. Clearly, if every vertex removed has degree exactly two, then $G$ is maximal 2-degenerate. 
The correctness of this method follows from the fact that no pair of degree two vertices are adjacent in a maximal 2-degenerate graph. 
This test is easily implemented in linear time via a pre-processing bucket sort of the vertices by degree and adjusting the ``bucket membership'' of each vertex with each vertex deletion. 
Thus, to recognize maximal L-graphs we will focus on the planar 2-degenerate graphs.

We now demonstrate a linear time test to determine whether $G$ has a 2-canonical order with a given base edge $e=v_1v_2$. We first the consider 2-degenerate orders of $G$ from a fixed base edge.

\begin{lemma}\label{lem:unique_precedence}
Let $G$ be planar 2-degenerate with an edge $e=v_1v_2$. For every vertex $v$ of $G$, in every 2-degenerate order starting from $e$, the neighbors of $v$ that precede $v$ are the same. Let $\overrightarrow{G_e}$ denote the orientation of $G$ according to the precedence order with base edge $e$. 
\end{lemma}
\begin{proof}
We prove this constructively. Clearly this is true for $v_1$ and $v_2$. Thus, we consider these vertices as marked. Now, for any vertex $v$ with exactly two marked neighbours, we know that these two vertices must precede $v$ in any 2-degenerate order. Notice that if some vertex has more than two marked neighbours then we know that this graph is not 2-degenerate since it contains a subgraph $H$ with more than $2|V(H)| -3$ edges. Similarly, if every unmarked vertex has less than two marked neighbors, then we know that $e$ is not the first edge of any 2-degenerate order. 
\end{proof}

Suppose we are given an edge $e = v_1v_2$ and need to determine whether $G$ has a 2-canonical order starting from $e$. We first construct a 2-degenerate order $\sigma$. If no such order exists, we reject $e$. Otherwise, by Lemma \ref{lem:unique_precedence}, we use $\sigma$ to construct $\overrightarrow{G_e}$. 

We initialize the L-representation ${\bf L} = \{L_{v_1}, L_{v_2}\}$ where $L_{v_1}$ is the ``top-most'' L-shape and $L_{v_2}$ is the ``right-most'' L-shape. We also initialize the admissible vertices $A$ as the vertices that could be added next according to $\overrightarrow{G}$ (i.e., $A$ contains the vertices adjacent to both $v_1$ and $v_2$). 

We now describe the main loop of our algorithm. Consider any admissible vertex $u_1$ and let $x$ and $y$ be $u_1's$ neighbors with $L_x, L_y \in {\bf L}$. Moreover, let $u_2, ..., u_k$ be the other admissible vertices adjacent to both $x$ and $y$. Notice that in order to add every $L_{u_i}$, we need an appropriate visibility between $L_x$ and $L_y$ in $\bf{L}$. However, we delay testing this until the end of the algorithm to save time. Observe the following properties of $u_1, \ldots, u_k$. The L-shapes corresponding to these vertices will be ``stacked'' on top of each other. This means that, if $e$ is the base edge of an L-representation of $G$, no pair $u_i$, $u_j$ can belong to the same connected component of $G \setminus \{x,y\}$. Thus, we let $H_i$ be the connected component of $G \setminus \{x,y\}$ which contains $u_i$. We now consider two cases. First, if (wlog) $H_1$ contains $v_1$, then $L_{u_1}$ must be ``lowest'' L-shape among $L_{u_1}, \ldots, L_{u_k}$ in any representation since it requires a path of L-shapes that reaches $L_{v_1}$ while avoiding $L_x$ and $L_y$. 
In particular, for each $i \in \{2, \ldots, k\}$,  we need $G_i = (G[H_i \cup \{x,y\}]$ together with the edge $xy$) to have an L-representation ${\bf L_i}$ with $xy$ as the base edge. Moreover, if $H_1$ does not contain $v_1$, we also need such an ${\bf L_1}$ for $G_1$. We recursively construct these $\bf{L_i}$'s then insert them into $\bf{L}$. 
If any recursive call fails, we know $e$ was not a good base edge for $G$. If $H_1$ contained $v_1$, we add the an L-shape for $u_1$ to $\bf{L}$, and update the admissible vertices with respect to $u_1$ (note: we don't need to update with respect to $u_2, \ldots, u_k$ since we have already processed their entire connected components). From here we repeat this main loop until we have exhausted all vertices or we have found a contradiction. After exhausting the vertices we check whether our constructed representation is correct. 

This completes the description of the algorithm and it is easy to see that it runs in polynomial time.
If one is careful, it can be implemented in linear time. In particular, we can augment $\overrightarrow{G}$ to capture the hierarchical representation of the connected components for each relevant separating pair $x,y$. Moreover, we know that we only need to consider a linear number of such separating pairs since each such pair corresponds to an admissible vertex. Finally, we can easily group the admissible vertices by their preceding neighbors when we add them to the admissible set. This will allow us to consider the related admissible vertices without searching through all admissible vertices. 
Thus, by applying the above test for every edge of $G$, we can recognize the 2-canonical graphs in quadratic time. 

\section{Equilateral L-representations and Related Representations}\label{sec:equilateral}

Every L-representation of $G$ with base edge $v_1v_2$ induces an edge labeling of $G$ with base edge $v_1v_2$, by orienting an edge $uv$ from $u$ to $v$ if an endpoint of $L_u$ is contained in the interior of $L_v$, and coloring it red (blue) if it is the top (right) endpoint of $L_u$. We say that two L-representations are \emph{equivalent} if they induce the same edge labeling. An L-shape is \emph{equilateral} if its horizontal and vertical leg are of the same length. An \emph{equilateral L-representation} is one with only equilateral L-shapes.
  
\begin{theorem}\label{th:equiL}
 Every L-representation has an equivalent equilateral L-representation.
\end{theorem}
\begin{proof}
 For a given L-representation with base edge $v_1v_2$, consider the induced edge labeling and fix one corresponding $2$-canonical order $(v_1,v_2,\ldots,v_n)$. We construct an equivalent L-representation with equilateral L-shapes along this $2$-canonical order, i.e., by a variant of the algorithm given in Lemma~\ref{lem:L-is-2-canonical}. We maintain the following invariant:
  
 \medskip
 \textit{{\bf Invariant:} 
  There is a line $\ell$ of slope $-1$ that intersects every segment of the outer staircase in an interior point.
 }
 \medskip

 In the beginning we fix the line $\ell$ arbitrarily -- say $\ell = \{(r,-r) \,|\, r \in \mathbb{R}\}$. We keep $\ell$ fixed throughout the entire construction. In the base case one can easily define the L-shapes $L_{v_1}$ and $L_{v_2}$ so that all four legs intersect $\ell$ in an interior point -- say $L_{v_1}$ and $L_{v_2}$ have top endpoint $(1,2)$ and $(3,-1)$, respectively, and right endpoint $(4,-1)$ and $(5,-3)$, respectively; see Fig.~\ref{fig:equi-basecase}. In general we have an L-representation of $G_i = G[v_1,\ldots,v_i]$ in which the invariant is maintained.

 \begin{figure}[t!]
  \centering
  \subfigure[]{
   \includegraphics[width=.24\textwidth]{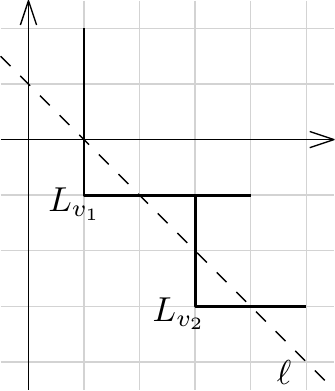}
   \label{fig:equi-basecase}
  }
  \hspace{1em}
  \subfigure[]{
   \includegraphics[width=.28\textwidth]{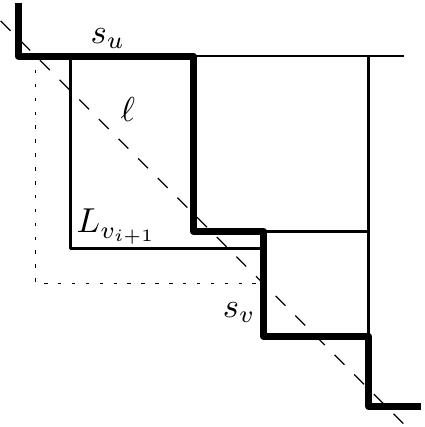}
   \label{fig:equi-step}
  }
  \hspace{0em}
  \subfigure[]{
 \includegraphics{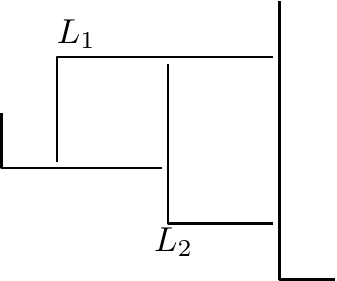}
 \label{fig:not-equilateral}
}
  \caption{\subref{fig:equi-basecase} The definition of $L_{v_1}$ and $L_{v_2}$. \subref{fig:equi-step} Introducing the L-shape for $v_{i+1}$ maintaining the invariant. \subref{fig:not-equilateral}
A contact L-representation with L-shapes in two different rotations without equivalent equilateral representation for both $L_1$ and $L_2$.}
  \label{fig:equi}
 \end{figure}

 Consider what happens when we insert a new L-shape for $v_{i+1}$. Let $u$ and $v$ be the two neighbors of $v_{i+1}$ in $G_{i+1}$. W.l.o.g. $u$ comes before $v$ when going counterclockwise around the outer face of $G_i$ starting at $v_1$. Let $s_u$ and $s_v$ be the horizontal segment and vertical segment of the outer staircase which are contained in $L_u$ and $L_v$, respectively. Note that by the invariant, if we would choose the points $\ell \cap s_u$ and $\ell \cap s_v$ as top and right endpoint of the newly inserted L-shape, then this would be equilateral. However, we do not insert $L_{v_{i+1}}$ exactly there as this would break the invariant. Instead, we insert a
 slightly smaller L-shape in such a way that the corresponding two new segments of the outer staircase intersect $\ell$ in the interior; see Fig.~\ref{fig:equi-step}.
\end{proof}

We remark that the equilateral L-representation constructed in Theorem~\ref{th:equiL} requires an exponential sized grid. Finding an equilateral L-representation on a polysize grid remains open. Further we remark that with more than one of the four possible rotations in an L-representation, it is no longer true that every L-representation has an equivalent equilateral one. Consider the L-representation in Fig.~\ref{fig:not-equilateral}: in every equivalent representation the horizontal leg of $L_1$ is longer than the horizontal leg of $L_2$ and the vertical leg of $L_1$ is shorter than the vertical leg of $L_2$. Thus $L_1$ and $L_2$ cannot be both equilateral.

For a maximally planar graph $G$ with Schnyder realizer $(S_1,S_2,S_n)$ and an inner vertex $v$ we define $\sigma_i(v)$ to be the outgoing neighbor of $v$ in $S_i$, $i=1,2,n$. For convenience, let $\sigma_n(v_1) = \sigma_n(v_2) = \sigma_n(v_n) = v_{n+1}$ for a dummy vertex $v_{n+1} \notin V(G)$.

\begin{definition}[cuboid representation]\label{def:cuboids}
 Let $G = (V,E)$ be a maximally planar graph, $(S_1,S_2,S_n)$ a Schnyder realizer of $G$, $\{L_v \,|\, v \neq v_n\}$ an L-representation of $G \setminus S_n$, and $h(v)$ a number for every vertex $v \in V \cup v_{n+1}$. For $v \neq v_n$ let $(x^r_v,y^r_v)$ and $(x^t_v,y^t_v)$ be the right and top endpoint of $L_v$, respectively. Define an L-shape $L_{v_n}$ with right endpoint $(x^r_{v_n},y^r_{v_n}) := (x^t_{v_2},y^r_{v_2})$ and top endpoint $(x^t_{v_n},y^t_{v_n}) := (x^t_{v_1},y^r_{v_1})$. Then for every $v \in V$ its \emph{cuboid} is defined as:
 $$Q_{v} := [x^t_v,x^r_v] \times [y^r_v,y^t_v] \times [h(\sigma_n(v)),h(v)]$$
\end{definition}

Note that for any $v$ the projection of $Q_v$ onto the $xy$-plane gives a rectangle, two sides of which form the L-shape $L_v$. The number $h(v)$ corresponds to the ``height'', i.e., $z$-coordinate, of the top side of the cuboid $Q_v$; see Fig.~\ref{fig:cuboids}. A \emph{cuboid representation} of a graph $G$ is a set of interior disjoint cuboids, one for each vertex, so that two cuboids intersect exactly if the corresponding vertices are adjacent in $G$. A cuboid representation is \emph{proper} if every non-empty intersection of two cuboids is a $2$-dimensional rectangle.


\begin{proposition}\label{prop:cuboids}
 The cuboids given by Definition~\ref{def:cuboids} form a cuboid representation of $G$ whenever $h(v_{n+1}) < h(v_n)$ and for every inner vertex $v$ of $G$ we have
 \begin{equation}
  h(\sigma_1(v)) \geq h(v) \quad\text{and}\quad h(\sigma_2(v)) \geq h(v) \quad\text{and}\quad h(\sigma_n(v)) < h(v).\label{eq:h-values}
 \end{equation}
Further, a non-degenerate L-representation implies a proper cuboid representation.
\end{proposition}

\begin{proof}
Note that conditions~\eqref{eq:h-values} imply that along the edges of $S_1 \cup S_2 \cup S_n^{-1}$ the $h$-values are non-decreasing.
%
 It is easy to show that the cuboids for the outer three vertices are mutually touching with proper side contacts. So let $uv$ be an inner edge of $G$. First assume $v = \sigma_i(u)$, i.e., $uv \in S_i$, for some $i \in \{1,2\}$. Looking at the L-representation we see that projecting $Q_u$ and $Q_v$ onto the $xy$-plane gives two rectangles with non-empty intersection or a proper side contact in the non-degenerate case, which is horizontal if $i=1$ and vertical if $i=2$. Projecting $Q_u$ and $Q_v$ onto the $z$-axis gives intervals $[h(\sigma_n(u)),h(u)]$ and $[h(\sigma_n(v)),h(v)]$, respectively. Since there is a directed path from $u$ to $\sigma_n(v)$ in $S_1 \cup S_2 \cup S_n^{-1}$ we get from~\eqref{eq:h-values} that $h(\sigma_n(v)) > h(u) \geq h(v)$. Thus $Q_u$ and $Q_v$ overlap non-trivially.
 
 Next assume $v = \sigma_n(u)$, i.e., $uv \in S_n$. Looking at the L-representation we see that projecting $Q_u$ and $Q_v$ onto the $xy$-plane gives two rectangles that intersect or overlap non-trivially in the non-degenerate case. Projecting $Q_u$ and $Q_v$ onto the $z$-axis gives intervals $[h(\sigma_n(u)),h(u)] = [h(v),h(u)]$ and $[h(\sigma_n(v)),h(v)]$, respectively. Thus $Q_u \cap Q_v \neq \emptyset$ or is a rectangle parallel to the $xy$-plane in the non-degenerate case.
  
 Finally let $u$ and $v$ be non-adjacent. If the rectangles defined by $L_u$ and $L_v$ do not overlap, i.e., can be separated by a horizontal or vertical line, then in $3$-space $Q_u$ and $Q_v$ are separated by a plane parallel to the $yz$-plane or $xz$-plane. If the rectangles do overlap, there is a path on at least two edges in $S_n$ starting and ending in $u$ and $v$, respectively. From~\eqref{eq:h-values} and the definition of the $z$-component of cuboids follows that $Q_u$ and $Q_v$ can separated by a plane parallel to the $xy$-plane. 
\end{proof}

\begin{theorem}\label{thm:square-based}
 Planar graphs have proper contact representation by square-based cuboids.
\end{theorem}
\begin{proof}
 As every planar graph is an induced subgraph of some maximally planar graph we may assume w.l.o.g. that $G = (V,E)$ is a maximally planar graph. We fix any Schnyder realizer $(S_1,S_2,S_n)$ of $G$, consider any non-degenerate equilateral L-re\-pre\-sen\-ta\-tion of $G \setminus S_n$, which exists by Theorem~\ref{th:equiL}. Further we let $(v_1,v_2,\ldots,v_n)$ be any canonical order of $G$ w.r.t. $S_1,S_2$ and define $h(v_i) = -i$ for $i=1,\ldots,n$ and $h(v_{n+1}) = -(n+1)$. Clearly,~\eqref{eq:h-values} holds for these $h$-values. Hence by Proposition~\ref{prop:cuboids} the cuboids given by Definition~\ref{def:cuboids} form a proper cuboid representation of $G$, and since the L-representation is equilateral every cuboid has a square base.
\end{proof}
 
We remark that a square-based cuboid representation can be found efficiently with an iterative approach, when the L-representation and the cuboids are defined along a single sweep of the chosen canonical order. This approach is illustrated in Fig.~\ref{fig:cuboids}.

\begin{figure}[t!]
 \centering
 \subfigure[]{
  \includegraphics{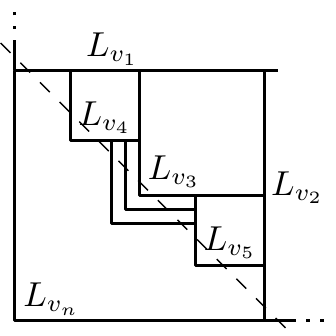}
  \label{fig:cuboids-L}
 }
 \subfigure[]{
  \includegraphics{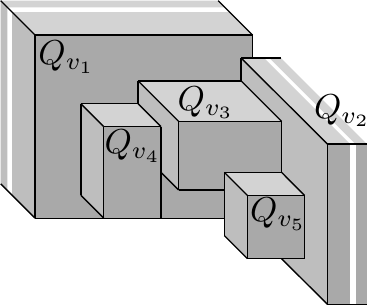}
  \label{fig:cuboids-3D-partial}
 }
 \subfigure[]{
  \includegraphics{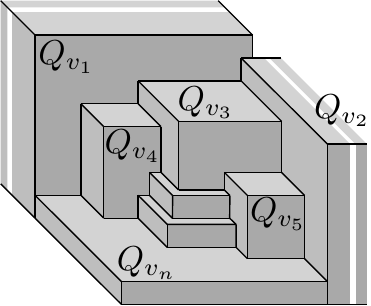}
  \label{fig:cuboids-3D-final}
 }
 \caption{\subref{fig:cuboids-L} An equilateral L-representation of $G \setminus S_n$ together with an L-shape for the vertex $v_n$. \subref{fig:cuboids-3D-partial}--\subref{fig:cuboids-3D-final} The cuboids can be defined along a canonical order w.r.t. $S_1,S_2$: The projection of each $Q_v$ onto the $xy$-plane is a rectangle spanned by $L_v$. The maximum and minimum $z$-coordinate of $Q_v$ is given by (the negative of) the index in the canonical order of $v$ and $\sigma_n(v)$, respectively.}
 \label{fig:cuboids}
\end{figure}

Next we address the question when the cuboids from Definition~\ref{def:cuboids} are actually cubes. This is clearly the case exactly if the chosen L-representation is equilateral and for every vertex $v$ we set $h(v) = h(\sigma_n(v)) + |L_v|$, where $|L_v|$ is the length of a leg of $L_v$. For a given equilateral L-representation we call this set of $h$-values the \emph{cubic heights}. We remark that in any L-representation we can choose the vertical leg of $L_{v_1}$ and the horizontal leg of $L_{v_2}$ (keeping the rest unchanged), so that $L_{v_1}$ and $L_{v_2}$ are equilateral. The cubic heights clearly satisfy $h(\sigma_n(v)) < h(v)$, but in general~\eqref{eq:h-values} is not satisfied and we are not guaranteed by Proposition~\ref{prop:cuboids} to obtain a cuboid representation. However, as we show next we can sometimes choose the equilateral L-representation (and implicitly the Schnyder realizer) more carefully so that~\eqref{eq:h-values} is satisfied for the cubic heights.

Consider a fixed L-representation and let $P$ be the set of all endpoints and bends of L-shapes. 
For a vertex $v$ let $\ell_v$ be the line through the top and right endpoint of $L_v$. A \emph{segment $s$ of an L-shape $L_v$} is a connected component of $L_v \setminus P$, i.e., $s \subset L_v$, each endpoint of $s$ is a point from $P$ and no further point from $P$ is contained in $s$. Let $C \subset P$ be the set of contact points between any two L-shapes. We call an L-representation \emph{Square-L, or SL-representation} if for every $p \in C$ the vertical segment whose right end is $p$ and the horizontal segment whose top end is $p$ have the same length; see Fig.~\ref{fig:T-equilateral}.

\begin{lemma}\label{lem:overconstrained}
 Consider a maximally planar graph $G$, a Schnyder realizer $(S_1,S_2,S_n)$, and an SL-representation of $G\setminus S_n$. Then for every $v \in V(G)$ the line $\ell_v$ has slope $-1$ and contains the bends of L-shapes corresponding to vertices $w$ with $\sigma_n(w) = v$.
\end{lemma}
\begin{proof}
 Consider any vertex $v \neq v_1,v_2$ and the corresponding L-shape $L_v$. Let $S_v$ be the staircase that connects the top and right endpoint of $L_v$ and contains the bends of L-shapes corresponding to vertices $w$ with $\sigma_n(w) = v$. If $s_1,\ldots,s_{2k}$ are the segments along $S_v$, then by assumption $s_{2i-1}$ and $s_{2i}$ are of the same length, $i=1,\ldots,k$. Equivalently, all bends on $S_v$ lie on $\ell_v$, and $\ell_v$ has slope $-1$.
\end{proof}

\begin{corollary}\label{cor:overconstrained}
 Let $\{L_v \;|\; v \in V\}$ be an SL-representation. Then it is equilateral and $\{\Delta_v := {\rm conv}(L_v) \;|\; v \in V\}$ is a homothetic triangle representation of $G$. Further, the cubic heights satisfy~\eqref{eq:h-values} and Proposition~\ref{prop:cuboids} yields a contact cube representation of $G$.
\end{corollary}

\begin{proof}
 Since $\ell_v$ has slope $-1$ and contains both endpoints of $L_v$, $L_v$ is equilateral for every $v$.
 For every vertex $v$ the sides of $\Delta_v$ are formed by $L_v$ and $\ell_v$. Since $\ell_v$ contains the bend of $L_w$ for every $w$ with $\sigma_n(w) = v$, $\Delta_w$ and $\Delta_v$ touch in a unique point. Moreover, any two triangles are interior disjoint and for any $v \neq w$ the top (right) corner of $\Delta_w$ touches $\Delta_v$ if $\sigma_1(w) = v$ ($\sigma_2(w) = v$). Thus $\{\Delta_v \;|\; v \in V(G)\}$ is a triangle representation of $G$ and since the L-shapes are homothetic, so are the triangles. See Fig.~\ref{fig:T-homothetics} for an example.

 Finally, we consider the cubic heights, i.e., $h(v) = h(\sigma_n(v)) + |L_v|$, and show that for any inner vertex $v$ of $G$~\eqref{eq:h-values} is satisfied. Consider for any inner vertex $v$ the path $P(v)$ in $S_n$ from $v$ to $v_n$. Then $h(v) = h(v_n) + \sum_{w \in V(P(v))} |L_w|$. On the other hand the distance between $\ell_{v_n}$ and $\ell_v$ is exactly $\frac{1}{\sqrt{2}}\sum_{w \in V(P(v)) \setminus v_n} |L_w|$. Since for $i=1,2$ we have that $\ell_v$ is closer to $\ell_{v_n}$ than $\ell_{\sigma_i(v)}$ it follows $h(\sigma_i(v)) \geq h(v)$. With $h(\sigma_n(v)) = h(v) - |L_v| < h(v)$ we conclude that~\eqref{eq:h-values} is satisfied.
\end{proof}

Not every L-representation has an equivalent SL-representation, since not every planar graph admits a homothetic triangle representation. But homothetic triangle representations exist for $4$-connected maximally planar graphs (Theorem~\ref{thm:homothetic-triangles}) and planar $3$-trees. 
Interestingly, only very few Schnyder realizers correspond to these representations -- in fact we do not know a graph that admits homothetic triangle representations for two distinct Schnyder realizers. Additionally, for a fixed Schnyder realizer there is at most one homothetic triangle representation.
Felsner and Francis~\cite{Felsner11} observe that from Theorem~\ref{thm:homothetic-triangles} one obtains a cube representation for every planar graph. However, the only known proof of Theorem~\ref{thm:homothetic-triangles} relies on Schramm's result~\cite{schramm2007combinatorically}, which does not give an efficient way to compute such a representation. 

Another approach for proving Theorem~\ref{thm:homothetic-triangles} was proposed by Felsner~\cite{Felsner11}. The idea is to guess a Schnyder realizer, compute a contact triangle representation, and set up a system of linear equations whose variables are the side lengths of triangles. The system has a unique solution and if it is non-negative it gives homothetic triangles. If the solution contains negative entries then from these one can read off a new Schnyder realizer and iterate. In practice, this always produces a homothetic triangle representation. However, there is no formal proof that this iterative procedure terminates.

Felsner's approach can be directly translated into our setting with L-representations. Guessing a Schnyder realizer we obtain an L-representation and an equation system whose variables are the lengths of segments. It has a unique solution and if it is non-negative we obtain an SL-representation. We believe that our interpretation may help to find homothetic triangle representations and hence cube representations efficiently. For example, the solution of the new equation system can be seen as two flows $f_h$ and $f_v$ in the visibility graph $G_h$ of horizontal and $G_v$ of vertical segments, respectively. Both, $G_h$ and $G_v$ are planar graphs, there is a vertex of $G_h$ in every face of $G_v$ and vice versa, and every edge of $G_h$ is crossed by a corresponding edge of $G_v$. However, $G_h$, $G_v$ are not a primal-dual pair of graphs; see Fig.~\ref{fig:flow-graphs}. The edges of $G_h$ and $G_v$ correspond to the horizontal and vertical segments, respectively. The solution to the equation system corresponds to an $s_h-t_h$ flow in $G_h$ and at the same time to an $s_v-t_v$ flow in $G_v$. A variable is positive if the flow through the corresponding edge in $G_h$ and $G_v$ goes bottom-up and left-to-right, respectively. A similar approach works for squarings of rectangular duals~\cite{FelsnerSurvey}, where $G_v$, $G_h$ is indeed a primal-dual pair of graphs. Considerations similar to those in~\cite{FelsnerSurvey} may give more insight to the problem.

 \begin{figure}[t!]
  \centering
  \subfigure[]{
   \includegraphics{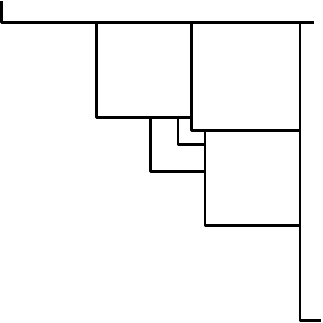}
   \label{fig:T-equilateral}
  }
  \subfigure[]{
   \includegraphics{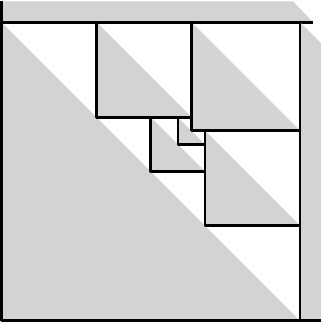}
   \label{fig:T-homothetics}
  }
  \subfigure[]{
   \includegraphics{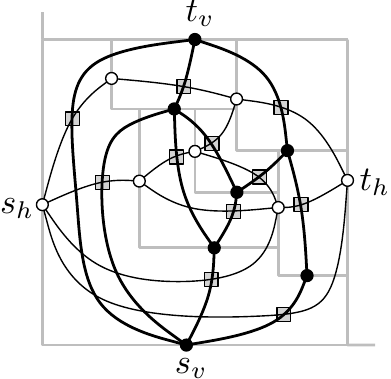}
   \label{fig:flow-graphs}
  }
  \caption{\subref{fig:T-equilateral} An SL-representation. \subref{fig:T-homothetics} The corresponding homothetic triangle representation. \subref{fig:flow-graphs} Graphs $G_h$, drawn thick on black vertices, and $G_v$, drawn thin on white vertices.
The gray boxes indicate which pairs of edges correspond to each other, i.e., the corresponding variables receive the same value in the $s_h-t_h$ flow in $G_h$ and the $s_v-t_v$ flow in $G_v$.}
  \label{fig:overconstrained}
 \end{figure}

\section{Conclusions and Open Problems}

We investigated L-graphs, provided a characterization, showed relations to Schnyder realizers and canonical orders, and described a recognition algorithm. Moreover, we showed that every L-representation can be transformed into an equivalent equilateral one, thus proving that every planar graph admits a proper contact representation with square-based cuboids, strengthening results by Thomassen~\cite{Thomassen86} and Bremner {\em et al.}~\cite{cubes12}. Finally, we showed that a more restrictive version of equilateral L-representations is equivalent to contact representations with homothetic triangles.
Many problems remain:
\begin{itemize}
 \item Characterizing contact L-graphs with L's in two or three rotations is open.
 \item Can L-graphs be recognized in linear time?
 \item Is there always an equilateral L-representation on a polynomial grid?
 \item Does every planar graph admit a {\em proper} contact representation with cubes?
 \item Can SL-representations help find homothetic triangle representations efficiently?
\end{itemize}

\medskip
\noindent{\bf\large Acknowledgments:} Initial work began at the Barbados Computational Geometry workshop in Feb.~2012, followed by work at the Berlin EuroGiga GraDR meeting in Nov.~2012. 
We thank organizers and participants for fruitful discussions and suggestions. We especially thank S.~Felsner, M.~Kaufmann, G.~Liotta, and T.~Mchedlidze for many discussions about several variants of this problem.

\bibliographystyle{abbrv}
{
\bibliography{laman}
}

%
%

\end{document}